\documentclass[12pt]{svmult}

\usepackage{amsmath,amssymb}
\usepackage{url}
\usepackage{algorithm}
\usepackage{algorithmic}

\begin{document}

\title*{Phase Transitions in Phase Retrieval}
\author{Dustin G.\ Mixon}
\institute{Department of Mathematics and Statistics\\Air Force Institute of Technology\\Wright-Patterson AFB, Ohio 45433}
\maketitle

\begin{abstract}
Consider a scenario in which an unknown signal is transformed by a known linear operator, and then the pointwise absolute value of the unknown output function is reported.
This scenario appears in several applications, and the goal is to recover the unknown signal -- this is called phase retrieval.
Phase retrieval has been a popular subject of research in the last few years, both in determining whether complete information is available with a given linear operator, and in finding efficient and stable phase retrieval algorithms in the cases where complete information is available.
Interestingly, there are a few ways to measure information completeness, and each way appears to be governed by a phase transition of sorts.
This chapter will survey the state of the art with some of these phase transitions, and identify a few open problems for further research.
\end{abstract}

\section{Introduction}

Various applications feature an inverse problem called \textit{phase retrieval}, in which one is given the pointwise absolute value of a known linear transformation of the desired signal.
Note that such information will never completely determine the unknown signal since, for example, negating the input will lead to the same output.
Indeed, the best one can do is recover $\{\omega x:|\omega|=1\}$ if $x$ is the unknown signal, but this global phase factor of ambiguity tends not to be an issue in application (for example, in some applications, the true signal is actually nonnegative everywhere, thereby removing all ambiguity).
What follows is a brief overview of some of the applications of phase retrieval:

\begin{itemize}
\item
\emph{Coherent diffractive imaging.}
This is a technique to image a nanoscale object by striking it with a highly coherent beam of X-rays to produce a diffraction pattern.
The diffraction pattern is the Fourier transform of the object, but only the intensity of the pattern can be physically measured (by counting photons in different regions)~\cite{BunkEtal:07,Harrison:93,MiaoISE:08,Millane:90}.
\item
\emph{Optics.}
This application enjoys various instances of phase retrieval:
(1) When imaging a star by a lens, one receives the pointwise absolute value of the Fourier transform of the desired pupil distribution~\cite{Walther:63}.
(2) For a high-resolution image, one can apply interferometric techniques to approximate the spatial coherence function (which is the Fourier transform of the desired object intensity), though the phase of this function is difficult to estimate accurately, so it is discarded~\cite{DaintyF:87}.
(3) Soon after NASA launched its Hubble Space Telescope, they discovered that the primary mirror in the telescope suffered from a large spherical aberration; the extent of this aberration was established by determining the pupil function from the intensity of its Fourier transform (the point spread function)~\cite{FienupMSS:93}.
\item
\emph{Quantum state tomography.}
When measuring a pure state (i.e., a unit vector) $x$ with a positive operator--valued measure of rank-$1$ elements $\{\varphi_n\varphi_n^*\}_{n=1}^N$ (i.e., the outer products of Parseval frame elements), the random outcome $Z$ of the measurement has a distribution given by $\operatorname{Pr}(Z=n)=|\langle x,\varphi_n\rangle|^2$.
As such, repeated measurements will produce an empirical estimate of the distribution, which is the pointwise absolute value (squared) of a linear transformation of the desired signal $x$~\cite{Finkelstein:04,FlammiaSC:05,HeinosaariMW:13}.
\item
\emph{Speech processing.}
One common method of speech signal denoising is to take the short-time Fourier transform (STFT) and perform a smoothing operation on the magnitudes of the coefficients~\cite{BalanCE:06}.
Instead of inverting the STFT with the noisy phases, one can recover the denoised version by phase retrieval~\cite{SunS:12}. 
\end{itemize}

Though there are many applications of phase retrieval, the task is often impossible; in particular, discarding the phases of the Fourier transform is not at all injective.
This fact has led many researchers to invoke a priori knowledge of the desired signal, since injectivity might be gotten by restricting to a smaller signal class.
For example, for the optics applications, the pupil distribution is only supported within the aperture of the optical system, and so this compact-support constraint combined with the intensity measurements might uniquely determine the desired signal.
Introducing such information has led to various ad hoc phase retrieval algorithms, and while they have found some success (e.g., in correcting the Hubble Space Telescope), such algorithms often fail to work unexpectedly.
Overall, this route has yet to produce algorithms with practical performance guarantees.

Thankfully, an alternative route was introduced in 2006 by Balan, Casazza and Edidin~\cite{BalanCE:06}:
Seek injectivity, not by restricting to a smaller signal class, but rather by designing a larger ensemble of measurement vectors.
Unbeknownst to Balan et al.\ at the time, this idea had already been in the air in the quantum community (for quantum state tomography~\cite{Finkelstein:04,FlammiaSC:05}), but posing the idea to the signal processing community led to a flurry of research in search of practical phase retrieval guarantees~\cite{AlexeevBFM:12,Balan:12,BalanBCE:09,BandeiraCM:13,CandesESV:13,CandesL:12,CandesLS:13,CandesSV:13,DemanetH:12,EldarM:12,Voroninski:12b,WaldspurgerAM:12}.
One popular method called \emph{PhaseLift} recasts phase retrieval as a semidefinite program~\cite{CandesESV:13,CandesL:12,CandesLS:13,CandesSV:13}, another called \emph{PhaseCut} reformulates it in terms of MaxCut~\cite{Voroninski:12b,WaldspurgerAM:12}, and yet another uses the polarization identity along with angular synchronization to quickly solve certain instances~\cite{AlexeevBFM:12,BandeiraCM:13}.
In this same line of research, a new methodology for coherent diffractive imaging emerged~\cite{CandesESV:13}:
Instead of taking a single exposure and attempting phase retrieval with possibly incomplete information, take multiple exposures of the same object with different masks or diffraction gratings.
Not only can such a process produce complete information, there are also provably efficient (and apparently stable) phase retrieval algorithms for this setting~\cite{BandeiraCM:13,CandesLS:13}.
Considering phase retrieval has a wide range of applications, it would be interesting to find other areas to apply this philosophy of taking more measurements to obtain injectivity.

Another effect of the paper~\cite{BalanCE:06} by Balan, Casazza and Edidin was the community's desire for a deeper understanding of injectivity for phase retrieval.
In particular, what are the conditions for injectivity, and how many measurements are required?
Understanding this will help in determining how many (and what type of) exposures are necessary for coherent diffractive imaging; also, such phase retrieval results can be directly interpreted as fundamental limits of quantum state tomography.
The purpose of this chapter is to survey several results along these lines, and to identity some remaining open problems.
Section 2 focuses on injectivity in both the real and complex cases, and then the third section considers a relaxed version of injectivity called \emph{almost injectivity}.
First, we discuss the notation we use throughout this chapter as well as some preliminaries:

\subsection{Notation and Preliminaries}

Given a collection of vectors $\Phi=\{\varphi_n\}_{n=1}^N$ in $V=\mathbb{R}^M$ or $\mathbb{C}^M$, we will identify such a collection with the $M\times N$ matrix whose columns form the collection.
Consider the intensity measurement process defined by 
\begin{equation*}
(\mathcal{A}(x))(n):=|\langle x,\varphi_n\rangle|^2.
\end{equation*}
Note that $\mathcal{A}(x)=\mathcal{A}(y)$ whenever $y=cx$ for some scalar $c$ of unit modulus.
As such, the mapping $\mathcal{A}\colon V\rightarrow\mathbb{R}^N$ is necessarily not injective.
To resolve this (technical) issue, throughout this chapter, we consider sets of the form $V/S$, where $V$ is a vector space and $S$ is a multiplicative subgroup of the field of scalars.
By this notation, we mean to identify vectors $x,y\in V$ for which there exists a scalar $c\in S$ such that $y=cx$; we write $y\equiv x\bmod S$ to convey this identification.
Most (but not all) of the time, $V/S$ is either $\mathbb{R}^M/\{\pm1\}$ or $\mathbb{C}^M/\mathbb{T}$ (here, $\mathbb{T}$ is the complex unit circle), and we view the intensity measurement process as a mapping $\mathcal{A}\colon V/S\rightarrow\mathbb{R}^N$; it is in this way that we will consider the measurement process to be injective.

As the title suggests, the focus of this chapter is phase transitions in phase retrieval.
As an example of a phase transition, consider what it takes for a collection of members of a vector space $V$ to span $V$.
Certainly, no collection of size less than the dimension of $V$ has a chance of spanning, and we expect most collections of size at least the dimension to span.
As such, we might say that the notion of spanning $V$ exhibits a phase transition at the dimension of $V$.
We make this definition explicit in the following:

\begin{definition}
Let $\mathrm{A}[\Phi;\mathbb{F}^{M\times N}]$ be a statement about a matrix $\Phi\in\mathbb{F}^{M\times N}$, and consider a function $f\colon\mathbb{N}\rightarrow\mathbb{N}$. 
We say $\mathrm{A}[\Phi;\mathbb{F}^{M\times N}]$ exhibits a \emph{phase transition} at $N=f(M)$ if for each $M\geq2$,
\begin{itemize}
\item[(a)]
$\mathrm{A}[\Phi;\mathbb{F}^{M\times N}]$ does not hold whenever $N<f(M)$, and
\item[(b)]
for each $N\geq f(M)$, there exists an open, dense subset $S\subseteq\mathbb{F}^{M\times N}$ such that $\mathrm{A}[\Phi;\mathbb{F}^{M\times N}]$ holds for every $\Phi\in S$.
\end{itemize}
\end{definition}

Based on experience, both parts (a) and (b) of a phase transition are established by first studying necessary and sufficient conditions for the property of interest $\mathrm{A}[\Phi;\mathbb{F}^{M\times N}]$.
For part (b) in particular, algebraic geometry consistently plays a key role.
Viewing $\Phi\in\mathbb{F}^{M\times N}$ as a point in real Euclidean space, consider the collection of $\Phi$'s for which $\mathrm{A}[\Phi;\mathbb{F}^{M\times N}]$ does not hold.
If there exists a real, nontrivial algebraic variety that contains these points (which is often the case), then every point $\Phi$ in the complement of the variety (which is open and dense in $\mathbb{F}^{M\times N}$) satisfies $\mathrm{A}[\Phi;\mathbb{F}^{M\times N}]$.
As such, for part (b), it suffices to identify the appropriate variety.

Throughout this chapter, we will continually follow a certain procedure for studying phase transitions.
Given a property $\mathrm{A}[\Phi;\mathbb{F}^{M\times N}]$, we start by studying various necessary and sufficient conditions for that property.
We then attempt to prove a phase transition $N=f(M)$ using the conditions available.
Later, we consider explicit constructions of $M\times f(M)$ matrices which satisfy $\mathrm{A}[\Phi;\mathbb{F}^{M\times f(M)}]$; these minimal constructions are certainly mathematically interesting, and they are also optimal measurement designs for applications like quantum state tomography.

\section{Injectivity}

In this section, we study the phase transition for injectivity.
As we will see, this phase transition is much better understood in the real case than in the complex case, and the distinction is rather interesting.
It is highly recommended that the reader enjoys the real case before venturing into the complex case.

\subsection{Injectivity in the Real Case}

We start by defining the important concepts of this subsection:

\begin{definition}
\
\begin{itemize}
\item[(a)]
$\mathrm{Inj}[\Phi;\mathbb{R}^{M\times N}]$ denotes the statement that $\mathcal{A}\colon\mathbb{R}^M/\{\pm1\}\rightarrow\mathbb{R}^N$ defined by $(\mathcal{A}(x))(n):=|\langle x,\varphi_n\rangle|^2$ is injective, where $\Phi=\{\varphi_n\}_{n=1}^N\subseteq\mathbb{R}^M$.
\item[(b)]
$\mathrm{CP}[\Phi;\mathbb{F}^{M\times N}]$ denotes the statement that $\Phi=\{\varphi_n\}_{n=1}^N\subseteq\mathbb{F}^M$ satisfies the \emph{complement property}: for every $S\subseteq\{1,\ldots,N\}$, either $\{\varphi_n\}_{n\in S}$ or $\{\varphi_n\}_{n\in S^\mathrm{c}}$ spans $\mathbb{F}^M$.
\end{itemize}
\end{definition}

Interestingly, the complement property characterizes injectivity in the real case.
Much of the work in this chapter was inspired by the proof of the following result:

\begin{theorem}[Theorem~2.8 in~\cite{BalanCE:06}]
\label{thm.complement property characterization}
$\mathrm{Inj}[\Phi;\mathbb{R}^{M\times N}]$ if and only if $\mathrm{CP}[\Phi;\mathbb{R}^{M\times N}]$.
In words, $\mathcal{A}$ is injective if and only if $\Phi$ satisfies the complement property.
\end{theorem}

\begin{proof}
We will prove both directions by obtaining the contrapositives.

($\Rightarrow$) 
Assume $\Phi$ does not satisfy the complement property.
Then there exists $S\subseteq\{1,\ldots,N\}$ such that neither $\{\varphi_n\}_{n\in S}$ nor $\{\varphi_n\}_{n\in S^\mathrm{c}}$ spans $\mathbb{R}^M$. 
This implies that there are nonzero vectors $u,v\in\mathbb{R}^M$ such that $\langle u,\varphi_n\rangle=0$ for all $n\in S$ and $\langle v,\varphi_n\rangle=0$ for all $n\in S^\mathrm{c}$.
For each $n$, we then have
\begin{align*}
|\langle u\pm v,\varphi_n\rangle|^2
&=|\langle u,\varphi_n\rangle|^2\pm2\operatorname{Re}\langle u,\varphi_n\rangle\overline{\langle v,\varphi_n\rangle}+|\langle v,\varphi_n\rangle|^2\\
&=|\langle u,\varphi_n\rangle|^2+|\langle v,\varphi_n\rangle|^2.
\end{align*}
Since $|\langle u+v,\varphi_n\rangle|^2=|\langle u-v,\varphi_n\rangle|^2$ for every $n$, we have $\mathcal{A}(u+v)=\mathcal{A}(u-v)$. 
Moreover, $u$ and $v$ are nonzero by assumption, and so $u+v\neq\pm(u-v)$.

($\Leftarrow$) 
Assume that $\mathcal{A}$ is not injective. 
Then there exist vectors $x,y\in\mathbb{R}^M$ such that $x\neq\pm y$ and $\mathcal{A}(x)=\mathcal{A}(y)$. 
Taking $S:=\{n:\langle x,\varphi_n\rangle=-\langle y,\varphi_n\rangle\}$, we have $\langle x+y,\varphi_n\rangle=0$ for every $n\in S$.
Otherwise when $n\in S^\mathrm{c}$, we have $\langle x,\varphi_n\rangle=\langle y,\varphi_n\rangle$ and so $\langle x-y,\varphi_n\rangle=0$.
Furthermore, both $x+y$ and $x-y$ are nontrivial since $x\neq\pm y$, and so neither $\{\varphi_n\}_{n\in S}$ nor $\{\varphi_n\}_{n\in S^\mathrm{c}}$ spans $\mathbb{R}^M$.
\qquad
\end{proof}

Having identified an equivalent condition (the complement property) to injectivity in the real case, we now use this condition to identify the phase transition.
First, we note that a spanning set for $\mathbb{R}^M$ must have size at least $M$.
As such, we know $\mathrm{CP}[\Phi;\mathbb{R}^{M\times N}]$ does not hold whenever $N<2M-1$, since taking $S$ to be the first $M-1$ members will leave $S^\mathrm{c}$ with $\leq M-1$ members.
This suggests a phase transition of $N=2M-1$, but it remains to prove part (b).
To get this, we first introduce the notion of full spark:
An $M\times N$ matrix with $M\leq N$ is said to be \emph{full spark} if every $M\times M$ submatrix is invertible.
Note that any full spark $\Phi$ with $N\geq 2M-1$ necessarily satisfies the complement property, since the larger of $S$ and $S^\mathrm{c}$ necessarily has at least $M$ elements, which necessarily span.
As such, it suffices to show that full spark matrices form an open and dense subset.
To this end, we note that the product of the determinants of all $M\times M$ submatrices forms a polynomial of the matrix entries whose zero set contains every $\Phi$ such that $\mathrm{CP}[\Phi;\mathbb{R}^{M\times N}]$ does not hold.
Moreover, this polynomial is nonzero since the Vandermonde matrix
\begin{equation*}
\left[
\begin{array}{cccc}
1&1&\cdots&1\\
1&2&\cdots&N\\
\vdots&\vdots&&\vdots\\
1^{M-1}&2^{M-1}&\cdots&N^{M-1}
\end{array}
\right]
\end{equation*}
is full spark (why?).
The complement of this polynomial's zero set is therefore open and dense, as desired.
This implies our first phase transition result:

\begin{theorem}[essentially proved in~\cite{BalanCE:06}]
$\mathrm{Inj}[\Phi;\mathbb{R}^{M\times N}]$ exhibits a phase transition at $N=2M-1$.
\end{theorem}

Now that we have identified the phase transition, we consider the minimal constructions, i.e., the $M\times (2M-1)$ real matrices which satisfy the complement property.
Here, we note that for every $S$ of size $M$, $S^\mathrm{c}$ has size $M-1$, meaning $S$ must index a spanning set.
As such, the matrices in this extreme case are precisely the $M\times (2M-1)$ full spark matrices.
For more information about full spark matrices, see~\cite{AlexeevCM:12}.

\subsection{Injectivity in the Complex Case}

In the previous subsection, we quickly identified a characterization of injectivity in the real case that enabled the phase transition of interest to be completely studied.
The complex case appears to be a bit more involved.
For example, the actual phase transition is the subject of an open conjecture, though there has been a lot of progress on this conjecture recently.
Since the complex case is so much more involved, this subsection is broken into different labeled parts, concerning necessary and sufficient conditions, the phase transition, and minimal constructions.

\subsubsection{Conditions for Injectivity in the Complex Case}

We begin by defining our symbol for injectivity in the complex case:

\begin{definition}
$\mathrm{Inj}[\Phi;\mathbb{C}^{M\times N}]$ denotes the statement that $\mathcal{A}\colon\mathbb{C}^M/\mathbb{T}\rightarrow\mathbb{R}^N$ defined by $(\mathcal{A}(x))(n):=|\langle x,\varphi_n\rangle|^2$ is injective, where $\Phi=\{\varphi_n\}_{n=1}^N\subseteq\mathbb{C}^M$.
\end{definition}

What follows is a characterization of injectivity in the complex case:

\begin{theorem}[Theorem~4 in~\cite{BandeiraCMN:13}]
\label{thm.complex injective}
Consider $\Phi=\{\varphi_n\}_{n=1}^N\subseteq\mathbb{C}^M$, and viewing $\{\varphi_n\varphi_n^* u\}_{n=1}^N$ as vectors in $\mathbb{R}^{2M}$, denote $S(u):=\operatorname{span}_\mathbb{R}\{\varphi_n\varphi_n^* u\}_{n=1}^N$.
Then the following are equivalent:
\begin{itemize}
\item[(a)] $\mathrm{Inj}[\Phi;\mathbb{C}^{M\times N}]$.
\item[(b)] $\operatorname{dim}S(u)\geq 2M-1$ for every $u\in\mathbb{C}^M\setminus\{0\}$.
\item[(c)] $S(u)=\operatorname{span}_\mathbb{R}\{\mathrm{i}u\}^\perp$ for every $u\in\mathbb{C}^M\setminus\{0\}$.
\end{itemize}
\end{theorem}

Before proving this theorem, note that unlike the characterization in the real case, it is not clear whether this characterization can be tested in finite time; instead of being a statement about all (finitely many) partitions of $\{1,\ldots,N\}$, this is a statement about all $u\in\mathbb{C}^M\setminus\{0\}$.
However, we can view this characterization as an analog to the real case in some sense:
In the real case, the complement property is equivalent to having $\operatorname{span}\{\varphi_n\varphi_n^* u\}_{n=1}^N=\mathbb{R}^M$ for all $u\in\mathbb{R}^M\setminus\{0\}$.
As the following proof makes precise, the fact that $\{\varphi_n\varphi_n^* u\}_{n=1}^N$ fails to span all of $\mathbb{R}^{2M}$ is rooted in the fact that more information is lost with phase in the complex case.
There is actually a nice differential geometric interpretation of this result, and we will discuss it after the proof:

\begin{proof}[Proof of Theorem~\ref{thm.complex injective}]
(a) $\Rightarrow$ (c): 
Suppose $\mathcal{A}$ is injective.
We need to show that $\{\varphi_n\varphi_n^* u\}_{n=1}^N$ spans the set of vectors orthogonal to $\mathrm{i}u$.
Here, orthogonality is with respect to the real inner product, which can be expressed as $\langle a,b\rangle_\mathbb{R}=\operatorname{Re}\langle a,b\rangle$.
Note that
\begin{equation*}
|\langle u\pm v,\varphi_n\rangle|^2
=|\langle u,\varphi_n\rangle|^2\pm2\operatorname{Re}\langle u,\varphi_n\rangle\langle \varphi_n,v\rangle+|\langle v,\varphi_n\rangle|^2,
\end{equation*}
and so subtraction gives
\begin{equation}
\label{eq.injectivity to orthogonality}
|\langle u+v,\varphi_n\rangle|^2-|\langle u-v,\varphi_n\rangle|^2
=4\operatorname{Re}\langle u,\varphi_n\rangle\langle \varphi_n,v\rangle
=4\langle \varphi_n\varphi_n^*u,v\rangle_\mathbb{R}.
\end{equation}
In particular, if the right-hand side of \eqref{eq.injectivity to orthogonality} is zero, then injectivity implies that there exists some $\omega$ of unit modulus such that $u+v=\omega(u-v)$.
Since $u\neq0$, we know $\omega\neq-1$, and so rearranging gives
\begin{equation*}
v
=-\frac{1-\omega}{1+\omega}u
=-\frac{(1-\omega)(1+\overline{\omega})}{|1+\omega|^2}u
=-\frac{2\operatorname{Im}\omega}{|1+\omega|^2}~\mathrm{i}u.
\end{equation*}
This means $S(u)^\perp\subseteq\operatorname{span}_\mathbb{R}\{\mathrm{i}u\}$.
To prove $\operatorname{span}_\mathbb{R}\{\mathrm{i}u\}\subseteq S(u)^\perp$, take $v=\alpha \mathrm{i}u$ for some $\alpha\in\mathbb{R}$ and define $\omega:=\frac{1+\alpha \mathrm{i}}{1-\alpha \mathrm{i}}$, which necessarily has unit modulus.
Then
\begin{equation*}
u+v
=u+\alpha \mathrm{i} u
=(1+\alpha \mathrm{i})u
=\frac{1+\alpha \mathrm{i}}{1-\alpha \mathrm{i}}(u-\alpha \mathrm{i}u)
=\omega(u-v).
\end{equation*}
Thus, the left-hand side of \eqref{eq.injectivity to orthogonality} is zero, meaning $v\in S(u)^\perp$.

(b) $\Leftrightarrow$ (c):
First, (b) immediately follows from (c).
For the other direction, note that $\mathrm{i}u$ is necessarily orthogonal to every $\varphi_n\varphi_n^* u$:
\begin{equation*}
\langle \varphi_n\varphi_n^* u,\mathrm{i}u\rangle_\mathbb{R}
=\operatorname{Re}\langle \varphi_n\varphi_n^* u,\mathrm{i}u\rangle
=\operatorname{Re}\langle u,\varphi_n\rangle\langle \varphi_n,\mathrm{i}u\rangle
=-\operatorname{Re}\mathrm{i}|\langle u,\varphi_n\rangle|^2
=0.
\end{equation*}
Thus, $\operatorname{span}_\mathbb{R}\{\mathrm{i}u\}\subseteq S(u)^\perp$, and by (b), $\operatorname{dim}S(u)^\perp\leq1$, both of which gives (c).

(c) $\Rightarrow$ (a):
This portion of the proof is inspired by Mukherjee's analysis in~\cite{Mukherjee:81}.
Suppose $\mathcal{A}(x)=\mathcal{A}(y)$.
If $x=y$, we are done.
Otherwise, $x-y\neq0$, and so we may apply (c) to $u=x-y$.
First, note that 
\begin{equation*}
\langle \varphi_n\varphi_n^*(x-y),x+y\rangle_\mathbb{R}
=\operatorname{Re}\langle \varphi_n\varphi_n^*(x-y),x+y\rangle
=\operatorname{Re}(x+y)^*\varphi_n\varphi_n^*(x-y),
\end{equation*}
and so expanding gives
\begin{align*}
\langle \varphi_n\varphi_n^*(x-y),x+y\rangle_\mathbb{R}
&=\operatorname{Re}\Big(|\varphi_n^*x|^2-x^*\varphi_n\varphi_n^*y+y^*\varphi_n\varphi_n^*x-|\varphi_n^*y|^2\Big)\\
&=\operatorname{Re}\Big(-x^*\varphi_n\varphi_n^*y+\overline{x^*\varphi_n\varphi_n^*y}\Big)
=0.
\end{align*}
Since $x+y\in S(x-y)^\perp=\operatorname{span}_\mathbb{R}\{\mathrm{i}(x-y)\}$, there exists $\alpha\in\mathbb{R}$ such that $x+y=\alpha \mathrm{i}(x-y)$, and so rearranging gives $y=\frac{1-\alpha \mathrm{i}}{1+\alpha \mathrm{i}}x$, meaning $y\equiv x\bmod\mathbb{T}$.
\qquad
\end{proof}

To better understand the above result, we will first develop a deeper understanding of the set $\mathbb{C}^M/\mathbb{T}$.
If we remove zero, this set happens to be something called a \emph{smooth manifold}, which means we can cover the set with overlapping patches, each with smooth coordinates, and with smooth coordinate transformations between the overlapping portions.
To see this, consider the patches defined by
\begin{equation*}
U_m:=\{Z\in(\mathbb{C}^M\setminus\{0\})/\mathbb{T}:Z_m\neq0\},
\qquad
m=1,\ldots,M.
\end{equation*}
We define the following coordinates over the patch $U_m$:
\begin{equation*}
(z_1,\ldots,z_M)=\frac{|Z_m|}{Z_m}(Z_1,\ldots,Z_M),
\end{equation*}
where $(Z_1,\ldots,Z_M)\in\mathbb{C}^M\setminus\{0\}$ denotes any representative of the corresponding point in $(\mathbb{C}^M\setminus\{0\})/\mathbb{T}$.
As such, each patch has its own homeomorphism to a set of coordinates, which is an open subset of $\mathbb{R}^{2M-1}$ (we lost a degree of freedom since the $m$th complex coordinate has no imaginary part).
If we denote the $m$th homeomorphism by $f_m\colon U_m\rightarrow\mathbb{R}^{2M-1}$, then it is not difficult to show that each $f_m$, as well as the transition maps $\tau_{m,m'}\colon f_m(U_m\cap U_{m'})\rightarrow f_{m'}(U_m\cap U_{m'})$ defined by $\tau_{m,m'}:=f_{m'}\circ f_m^{-1}$, are all smooth.

So we now understand that $(\mathbb{C}^M\setminus\{0\})/\mathbb{T}$ is a smooth manifold with $2M-1$ real dimensions.
If we consider the function $\mathcal{A}$ over $(\mathbb{C}^M\setminus\{0\})/\mathbb{T}$, we can take its derivative at a given point $u$ in terms of some chosen local coordinates.
This amounts to taking the Jacobian at $u\in U_m$, whose rows are 
$\{2\varphi_n\varphi_n^* u\}_{n=1}^N$ as vectors in $\mathbb{R}^{2M}$, but with the column corresponding to the $m$th imaginary component removed.
With this in mind, and using ideas from the proof of Lemma~22 in~\cite{BandeiraCMN:13}, Theorem~\ref{thm.complex injective} can be reinterpreted as follows:

\begin{theorem}
$\mathcal{A}$ is injective if and only if the derivative of $\mathcal{A}$ is injective at every point in $(\mathbb{C}^M\setminus\{0\})/\mathbb{T}$.
\end{theorem}

This says quite a bit about the intensity measurement mapping $\mathcal{A}$.
Indeed, one can imagine a smooth mapping of a circle to a figure-eight curve, in which the derivative is injective at every point of the circle, but the mapping is certainly not injective.
One could identify injectivity of the derivative as a local form of injectivity, and so it is rather surprising to have this be equivalent to the traditional (global) form of injectivity. 
Unfortunately, it is not clear what one can glean from such a feature of $\mathcal{A}$.
Indeed, the above theorems leave a lot to be desired; compared to the complement property in the real case, it is still unclear what it takes for a complex ensemble to yield injective intensity measurements.
While in pursuit of a more clear understanding, the following bizarre characterization was stumbled upon:
A complex ensemble yields injective intensity measurements precisely when it yields injective \emph{phase-only} measurements (in some sense).
This is made more precise in the following theorem statement:

\begin{theorem}[Theorem~5 in~\cite{BandeiraCMN:13}]
\label{thm.complex phase only}
Consider $\Phi=\{\varphi_n\}_{n=1}^N\subseteq\mathbb{C}^M$ and the mapping $\mathcal{A}\colon\mathbb{C}^M/\mathbb{T}\rightarrow\mathbb{R}^N$ defined by $(\mathcal{A}(x))(n):=|\langle x,\varphi_n\rangle|^2$.
Then $\mathcal{A}$ is injective if and only if the following statement holds: 
If for every $n=1,\ldots,N$, either $\operatorname{arg}(\langle x,\varphi_n\rangle^2)=\operatorname{arg}(\langle y,\varphi_n\rangle^2)$ or one of the sides is not well-defined, then $x=0$, $y=0$, or $y\equiv x\bmod\mathbb{R}\setminus\{0\}$. 
\end{theorem}

\begin{proof}
By Theorem~\ref{thm.complex injective}, $\mathcal{A}$ is injective if and only if
\begin{equation}
\label{eq.injective equivalence 1}
\forall x\in\mathbb{C}^M\setminus\{0\}, \qquad\operatorname{span}_\mathbb{R}\{\varphi_n\varphi_n^*x\}_{n=1}^N=\operatorname{span}_\mathbb{R}\{\mathrm{i}x\}^{\perp}.
\end{equation}
Taking orthogonal complements of both sides, note that regardless of $x\in\mathbb{C}^M\setminus\{0\}$, we know $\operatorname{span}_\mathbb{R}\{\mathrm{i}x\}$ is necessarily a subset of $(\operatorname{span}_\mathbb{R}\{\varphi_n\varphi_n^*x\}_{n=1}^N)^\perp$, and so \eqref{eq.injective equivalence 1} is equivalent to
\begin{align*}
\forall x\in\mathbb{C}^M\setminus\{0\},\qquad&\operatorname{Re}\langle\varphi_n\varphi_n^*x,\mathrm{i}y\rangle=0\quad\forall n=1,\ldots,N\\
&\qquad\Longrightarrow\quad y=0 \text{ or } y\equiv x\bmod\mathbb{R}\setminus\{0\}.
\end{align*}
Thus, we need to determine when $\operatorname{Im}\langle x,\varphi_n\rangle\overline{\langle y,\varphi_n\rangle}=\operatorname{Re}\langle\varphi_n\varphi_n^*x,\mathrm{i}y\rangle=0$.
We claim that this is true if and only if $\operatorname{arg}(\langle x,\varphi_n\rangle^2)=\operatorname{arg}(\langle y,\varphi_n\rangle^2)$ or one of the sides is not well-defined.
To see this, we substitute $a:=\langle x,\varphi_n\rangle$ and $b:=\langle y,\varphi_n\rangle$.
Then to complete the proof, it suffices to show that $\operatorname{Im}a\overline{b}=0$ if and only if $\operatorname{arg}(a^2)=\operatorname{arg}(b^2)$, $a=0$, or $b=0$.

($\Leftarrow$)
If either $a$ or $b$ is zero, the result is immediate.
Otherwise, if $2\operatorname{arg}(a)=\operatorname{arg}(a^2)=\operatorname{arg}(b^2)=2\operatorname{arg}(b)$, then $2\pi$ divides $2(\operatorname{arg}(a)-\operatorname{arg}(b))$, and so $\operatorname{arg}(a\overline{b})=\operatorname{arg}(a)-\operatorname{arg}(b)$ is a multiple of $\pi$.
This implies that $a\overline{b}\in\mathbb{R}$, and so $\operatorname{Im}a\overline{b}=0$.

($\Rightarrow$)
Suppose $\operatorname{Im}a\overline{b}=0$.
Taking the polar decompositions $a=re^{\mathrm{i}\theta}$ and $b=se^{\mathrm{i}\phi}$, we equivalently have that $rs\sin{(\theta-\phi)}=0$.
Certainly, this can occur whenever $r$ or $s$ is zero, i.e., $a=0$ or $b=0$.
Otherwise, a difference formula then gives $\sin{\theta}\cos{\phi}=\cos{\theta}\sin{\phi}$.
From this, we know that if $\theta$ is an integer multiple of $\pi/2$, then $\phi$ is as well, and vice versa, in which case $\operatorname{arg}(a^2)=2\operatorname{arg}(a)=\pi=2\operatorname{arg}(b)=\operatorname{arg}(b^2)$.
Else, we can divide both sides by $\cos{\theta}\cos{\phi}$ to obtain $\tan{\theta}=\tan{\phi}$, from which it is evident that $\theta\equiv\phi\bmod\pi$, and so $\operatorname{arg}(a^2)=2\operatorname{arg}(a)=2\operatorname{arg}(b)=\operatorname{arg}(b^2)$.
\qquad
\end{proof}

The notion of injective phase-only measurements appears similar to the notion of parallel rigidity in certain location estimation problems (for example, see~\cite{OzyesilSB:13} and references therein).
It would be interesting to further investigate this relationship, but at the very least, it is rather striking that injectivity in one setting is equivalent to injectivity in the other.
In~\cite{BandeiraCMN:13}, this equivalence is used to prove that the complement property is necessary for injectivity in the complex case.
Contrary to what is claimed in~\cite{BalanCE:06}, the first part of the proof of Theorem~\ref{thm.complement property characterization} does not suffice: 
It demonstrates that $u+v\neq\pm(u-v)$, but fails to establish that $u+v\not\equiv u-v\bmod\mathbb{T}$; for instance, it could very well be the case that $u+v=\mathrm{i}(u-v)$, and so injectivity would not be violated in the complex case.
Overall, the complement property is necessary but not sufficient for injectivity.  
To see that it is not sufficient, consider measurement vectors $(1,0)$, $(0,1)$ and $(1,1)$.
These certainly satisfy the complement property, but $\mathcal{A}((1,\mathrm{i}))=(1,1,2)=\mathcal{A}((1,-\mathrm{i}))$, despite the fact that $(1,\mathrm{i})\not\equiv(1,-\mathrm{i})\bmod\mathbb{T}$; in general, real measurement vectors fail to yield injective intensity measurements in the complex setting since they do not distinguish complex conjugates.

The theorem that follows provides one last characterization of injectivity in the complex case, and it will play a key role in our understanding of the phase transition.
Before stating the result, define the real $M^2$-dimensional space $\mathbb{H}^{M\times M}$ of self-adjoint $M\times M$ matrices; note that this is not a vector space over complex scalars since the diagonal of a self-adjoint matrix must be real.
Given an ensemble of measurement vectors $\{\varphi_n\}_{n=1}^N\subseteq\mathbb{C}^M$, 
define the \emph{super analysis operator} $\mathbf{A}\colon\mathbb{H}^{M\times M}\rightarrow\mathbb{R}^N$ by $(\mathbf{A}H)(n)=\langle H,\varphi_n\varphi_n^*\rangle_\mathrm{HS}$; here, $\langle\cdot,\cdot\rangle_\mathrm{HS}$ denotes the Hilbert-Schmidt inner product, which induces the Frobenius matrix norm.
Note that $\mathbf{A}$ is a linear operator, and yet 
\begin{align*}
(\mathbf{A}xx^*)(n)
&=\langle xx^*,\varphi_n\varphi_n^*\rangle_\mathrm{HS}
=\operatorname{Tr}[\varphi_n\varphi_n^*xx^*]\\
&=\operatorname{Tr}[\varphi_n^*xx^*\varphi_n]
=\varphi_n^*xx^*\varphi_n
=|\langle x,\varphi_n\rangle|^2
=(\mathcal{A}(x))(n).
\end{align*}
In words, the class of vectors identified with $x$ modulo $\mathbb{T}$ can be ``lifted'' to $xx^*$, thereby linearizing the intensity measurement process at the price of squaring the dimension of the vector space of interest; this identification has been exploited by some of the most noteworthy strides in modern phase retrieval~\cite{BalanBCE:09,CandesSV:13}.
As the following lemma shows, this identification can also be used to characterize injectivity:

\begin{theorem}[Lemma~9 in~\cite{BandeiraCMN:13}, cf.\ Corollary~1 in~\cite{HeinosaariMW:13}]
\label{lem.not injective rank 2}
$\mathcal{A}$ is not injective if and only if there exists a matrix of rank $1$ or $2$ in the null space of $\mathbf{A}$.
\end{theorem}

\begin{proof}
($\Rightarrow$)
If $\mathcal{A}$ is not injective, then there exist $x,y\in\mathbb{C}^M/\mathbb{T}$ with $x\not\equiv y\bmod\mathbb{T}$ such that $\mathcal{A}(x)=\mathcal{A}(y)$.
That is, $\mathbf{A}xx^*=\mathbf{A}yy^*$, and so $xx^*-yy^*$ is in the null space of $\mathbf{A}$.

($\Leftarrow$)
First, suppose there is a rank-$1$ matrix $H$ in the null space of $\mathbf{A}$.
Then there exists $x\in\mathbb{C}^M$ such that $H=xx^*$ and $(\mathcal{A}(x))(n)=(\mathbf{A}xx^*)(n)=0=(\mathcal{A}(0))(n)$.
But $x\not\equiv0\bmod\mathbb{T}$, and so $\mathcal{A}$ is not injective.
Now suppose there is a rank-$2$ matrix $H$ in the null space of $\mathbf{A}$.
Then by the spectral theorem, there are orthonormal $u_1,u_2\in\mathbb{C}^M$ and nonzero $\lambda_1\geq\lambda_2$ such that $H=\lambda_1u_1u_1^*+\lambda_2u_2u_2^*$.
Since $H$ is in the null space of $\mathbf{A}$, the following holds for every $n$:
\begin{align}
\nonumber
0
&=\langle H,\varphi_n\varphi_n^*\rangle_\mathrm{HS}\\
\label{eq.rank 2 injective}
&=\langle \lambda_1u_1u_1^*+\lambda_2u_2u_2^*,\varphi_n\varphi_n^*\rangle_\mathrm{HS}
=\lambda_1|\langle u_1,\varphi_n\rangle|^2+\lambda_2|\langle u_2,\varphi_n\rangle|^2.
\end{align}
Taking $x:=|\lambda_1|^{1/2}u_1$ and $y:=|\lambda_2|^{1/2}u_2$, note that $y\not\equiv x\bmod\mathbb{T}$ since they are nonzero and orthogonal.
We claim that $\mathcal{A}(x)=\mathcal{A}(y)$, which would complete the proof.
If $\lambda_1$ and $\lambda_2$ have the same sign, then by \eqref{eq.rank 2 injective}, $|\langle x,\varphi_n\rangle|^2+|\langle y,\varphi_n\rangle|^2=0$ for every $n$, meaning $|\langle x,\varphi_n\rangle|^2=0=|\langle y,\varphi_n\rangle|^2$.
Otherwise, $\lambda_1>0>\lambda_2$, and so $xx^*-yy^*=\lambda_1u_1u_1^*+\lambda_2u_2u_2^*=A$ is in the null space of $\mathbf{A}$, meaning $\mathcal{A}(x)=\mathbf{A}xx^*=\mathbf{A}yy^*=\mathcal{A}(y)$.
\qquad
\end{proof}

Of the three characterizations of injectivity in the complex case that we provided, this is by far the easiest to grasp, and perhaps due to its simplicity, much of our current understanding of the phase transition is based on this one.
Still, comparing to our understanding in the real case helps to identify areas for improvement.
For example, it is unclear how to test whether matrices of rank 1 or 2 lie in the null space of an arbitrary super analysis operator.
Indeed, we have yet to find a ``good'' sufficient condition for injectivity in the complex case like the complement property or full spark provide in the real case.

\subsubsection{The Phase Transition for Injectivity in the Complex Case}

At this point, we wish to study the phase transition (presuming it exists) for $\mathrm{Inj}[\Phi;\mathbb{C}^{M\times N}]$.
To this end, we introduce the following subproblem (i.e., part (a) of the phase transition):

\begin{problem}
\label{prob.n star}
For any dimension $M$, what is the smallest number $N^*(M)$ of injective intensity measurements?
\end{problem}

Interestingly, this problem has some history in the quantum mechanics literature.
For example, \cite{Vogt:78} presents \emph{Wright's conjecture} that three observables suffice to uniquely determine any pure state.
In phase retrieval parlance, the conjecture states that there exist unitary matrices $U_1$, $U_2$ and $U_3$ such that $\Phi=[U_1~U_2~U_3]$ yields injective intensity measurements.
Note that Wright's conjecture actually implies that $N^*(M)\leq 3M-2$; indeed, $U_1$ determines the norm (squared) of the signal, rendering the last column of both $U_2$ and $U_3$ unnecessary.
Finkelstein~\cite{Finkelstein:04} later proved that $N^*(M)\geq 3M-2$; combined with Wright's conjecture, this led many to believe that $N^*(M)=3M-2$ (for example, see~\cite{CandesESV:13}).
However, both this and Wright's conjecture were recently disproved in~\cite{HeinosaariMW:13}, in which Heinosaari, Mazzarella and Wolf invoked embedding theorems from differential geometry to prove that 
\begin{equation}
\label{eq.HMW lower bound}
N^*(M)\geq
\left\{
\begin{array}{ll}
4M-2\alpha(M-1)-3 &\mbox{for all } M\\
4M-2\alpha(M-1)-2 &\mbox{if } M\mbox{ is odd, }\alpha(M-1)\equiv 2\bmod 4\\
4M-2\alpha(M-1)-1 &\mbox{if } M\mbox{ is odd, }\alpha(M-1)\equiv 3\bmod 4
\end{array}
\right.
\end{equation}
where $\alpha(M-1)\leq\log_2(M)$ is the number of $1$'s in the binary representation of $M-1$; apparently, this result had previously appeared in~\cite{Moroz:84,MorozP:94} as well.
By comparison, Balan, Casazza and Edidin~\cite{BalanCE:06} proved that $N^*(M)\leq4M-2$, and so we at least have the asymptotic expression $N^*(M)=(4+o(1))M$.

At this point, we should clarify some intuition for $N^*(M)$ by explaining the nature of these best known lower and upper bounds.
First, the lower bound~\eqref{eq.HMW lower bound} follows from an older result that complex projective space $\mathbb{C}\mathbf{P}^n$ does not smoothly embed into $\mathbb{R}^{4n-2\alpha(n)}$ (and other slight refinements which depend on $n$); this is due to Mayer~\cite{Mayer:65}, but we highly recommend James's survey on the topic~\cite{James:71}.
To prove~\eqref{eq.HMW lower bound} from this, suppose $\mathcal{A}\colon\mathbb{C}^M/\mathbb{T}\rightarrow\mathbb{R}^N$ were injective.
Then $\mathcal{E}$ defined by $\mathcal{E}(x):=\mathcal{A}(x)/\|x\|^2$ embeds $\mathbb{C}\mathbf{P}^{M-1}$ into $\mathbb{R}^N$, and as Heinosaari et al.\ show, the embedding is necessarily smooth; considering $\mathcal{A}(x)$ is made up of rather simple polynomials, the fact that $\mathcal{E}$ is smooth should not come as a surprise.
As such, the nonembedding result produces the best known lower bound.
To evaluate this bound, first note that Milgram~\cite{Milgram:67} constructs an embedding of $\mathbb{C}\mathbf{P}^n$ into $\mathbb{R}^{4n-\alpha(n)+1}$, establishing the importance of the $\alpha(n)$ term, but the constructed embedding does not correspond to an intensity measurement process.
In order to relate these embedding results to our problem, consider the real case: 
It is known that for odd $n\geq7$, real projective space $\mathbb{R}\mathbf{P}^n$ smoothly embeds into $\mathbb{R}^{2n-\alpha(n)+1}$~\cite{Steer:70}, which means the analogous lower bound for the real case would necessarily be smaller than $2(M-1)-\alpha(M-1)+1=2M-\alpha(M-1)-1<2M-1$.
This indicates that the $\alpha(M-1)$ term in~\eqref{eq.HMW lower bound} might be an artifact of the proof technique, rather than of $N^*(M)$.

We now consider our previous analysis of injectivity to help guide intuition about the possible phase transition.
Theorem~\ref{lem.not injective rank 2} indicates that we want the null space of $\mathbf{A}$ to avoid nonzero matrices of rank $\leq 2$.
Intuitively, this is easier when the ``dimension'' of this set of matrices is small.
To get some idea of this dimension, let's count real degrees of freedom.
By the spectral theorem, almost every matrix in $\mathbb{H}^{M\times M}$ of rank $\leq 2$ can be uniquely expressed as $\lambda_1u_1u_1^*+\lambda_2u_2u_2^*$ with $\lambda_1\leq\lambda_2$.
Here, $(\lambda_1,\lambda_2)$ has two degrees of freedom.
Next, $u_1$ can be any vector in $\mathbb{C}^M$, except its norm must be $1$.
Also, since $u_1$ is only unique up to global phase, we take its first entry to be nonnegative without loss of generality.
Given the norm and phase constraints, $u_1$ has a total of $2M-2$ real degrees of freedom.
Finally, $u_2$ has the same norm and phase constraints, but it must also be orthogonal to $u_1$, that is, $\operatorname{Re}\langle u_2,u_1\rangle=\operatorname{Im}\langle u_2,u_1\rangle=0$.
As such, $u_2$ has $2M-4$ real degrees of freedom.
All together, we can expect the set of matrices in question to have $2+(2M-2)+(2M-4)=4M-4$ real dimensions.

If the set $S$ of matrices of rank $\leq 2$ formed a subspace of $\mathbb{H}^{M\times M}$ (it doesn't), then we could expect the null space of $\mathbf{A}$ to intersect that subspace nontrivially whenever $\dim\operatorname{null}(\mathbf{A})+(4M-4)>\dim(\mathbb{H}^{M\times M})=M^2$.
By the rank-nullity theorem, this would indicate that injectivity requires 
\begin{equation}
\label{eq.4M-4 reasoning}
N
\geq\operatorname{rank}(\mathbf{A})
=M^2-\dim\operatorname{null}(\mathbf{A})
\geq 4M-4.
\end{equation}
Of course, this logic is not technically valid since $S$ is not a subspace.
It is, however, a special kind of set: a real projective variety.
To see this, let's first show that it is a real algebraic variety, specifically, the set of members of $\mathbb{H}^{M\times M}$ for which all $3\times 3$ minors are zero.
Of course, every member of $S$ has this minor property.
Next, we show that members of $S$ are the only matrices with this property: 
If the rank of a given matrix is $\geq3$, then it has an $M\times 3$ submatrix of linearly independent columns, and since the rank of its transpose is also $\geq3$, this $M\times 3$ submatrix must have $3$ linearly independent rows, thereby implicating a full-rank $3\times 3$ submatrix.
This variety is said to be projective because it is closed under scalar multiplication. 
If $S$ were a projective variety over an algebraically closed field (it's not), then the projective dimension theorem (Theorem 7.2 of~\cite{Hartshorne:77}) says that $S$ intersects $\operatorname{null}(\mathbf{A})$ nontrivially whenever the dimensions are large enough: $\dim\operatorname{null}(\mathbf{A})+\dim S>\dim\mathbb{H}^{M\times M}$, thereby implying that injectivity requires \eqref{eq.4M-4 reasoning}.
Unfortunately, this theorem is not valid when the field is $\mathbb{R}$;
for example, the cone defined by $x^2+y^2-z^2=0$ in $\mathbb{R}^3$ is a projective variety of dimension $2$, but its intersection with the $2$-dimensional $xy$-plane is trivial, despite the fact that $2+2>3$.

In the absence of a proof, we pose the natural conjecture:

\begin{conjecture}[The $4M-4$ conjecture]
$\mathrm{Inj}[\Phi;\mathbb{C}^{M\times N}]$ exhibits a phase transition at $N=4M-4$.
\end{conjecture}

As incremental progress toward solving the $4M-4$ conjecture, we offer the following result:

\begin{theorem}[Theorem~10 in~\cite{BandeiraCMN:13}]
\label{thm.4M4 true 2}
The $4M-4$ conjecture is true when $M=2$.
\end{theorem}

In this case, injectivity is equivalent to $\mathbf{A}$ having a trivial null space by Theorem~\ref{lem.not injective rank 2}, meaning $\mathbf{A}$ must have rank $M^2=4=4M-4$ for injectivity, implying part (a).
For $N=4M-4$, $\mathbf{A}$ has a square matrix representation, and so injectivity is equivalent to having $\det\mathbf{A}\neq0$.
As such, part (b) is proved by considering the real algebraic variety $V:=\{\mathbf{A}:\mathrm{Re}\det\mathbf{A}=\mathrm{Im}\det\mathbf{A}=0\}$ and showing that $V^\mathrm{c}$ is nonempty.

We can also prove the $M=3$ case, but we first introduce Algorithm~\ref{algorithm}, namely the \emph{HMW test} for injectivity; this is named after Heinosaari, Mazarella and Wolf, who implicitly introduce this algorithm in their paper~\cite{HeinosaariMW:13}.

\begin{algorithm}[h]
\label{algorithm}
\caption{The HMW test for injectivity when $M=3$}
\textbf{Input:} Measurement vectors $\{\varphi_n\}_{n=1}^N\subseteq\mathbb{C}^3$\\
\textbf{Output:} Whether $\mathcal{A}$ is injective
\begin{algorithmic}
\STATE Define $\mathbf{A}\colon\mathbb{H}^{3\times 3}\rightarrow\mathbb{R}^N$ such that $\mathbf{A}H=\{\langle H,\varphi_n\varphi_n^*\rangle_\textrm{HS}\}_{n=1}^N$
\IF{$\operatorname{dim}\operatorname{null}(\mathbf{A})=0$}
\STATE ``INJECTIVE'' \hfill \tiny{\COMMENT{if $\mathbf{A}$ is injective, then $\mathcal{A}$ is injective}}
\ELSE
\STATE Pick $H\in\operatorname{null}(\mathbf{A})$, $H\neq0$
\IF{$\operatorname{dim}\operatorname{null}(\mathbf{A})=1$ and $\det(H)\neq0$}
\STATE ``INJECTIVE'' \hfill \tiny{\COMMENT{if $\mathbf{A}$ only maps nonsingular matrices to zero, then $\mathcal{A}$ is injective}}
\ELSE
\STATE ``NOT INJECTIVE'' \hfill \tiny{\COMMENT{in the remaining case, $\mathbf{A}$ maps differences of rank-$1$ matrices to zero}}
\ENDIF
\ENDIF
\end{algorithmic}
\end{algorithm}

\begin{theorem}[Theorem~11 in~\cite{BandeiraCMN:13}, cf.\ Proposition~6 in~\cite{HeinosaariMW:13}]
When $M=3$, the HMW test correctly determines whether $\mathcal{A}$ is injective.
\end{theorem}

\begin{proof}
First, if $\mathbf{A}$ is injective, then $\mathcal{A}(x)=\mathbf{A}xx^*=\mathbf{A}yy^*=\mathcal{A}(y)$ if and only if $xx^*=yy^*$, i.e., $y\equiv x\bmod\mathbb{T}$.
Next, suppose $\mathbf{A}$ has a $1$-dimensional null space.
Then Lemma~\ref{lem.not injective rank 2} gives that $\mathcal{A}$ is injective if and only if the null space of $\mathbf{A}$ is spanned by a matrix of full rank.
Finally, if the dimension of the null space is $2$ or more, then there exist linearly independent (nonzero) matrices $A$ and $B$ in this null space.
If $\det{(A)}=0$, then it must have rank $1$ or $2$, and so Lemma~\ref{lem.not injective rank 2} gives that $\mathcal{A}$ is not injective.
Otherwise, consider the map
\begin{equation*}
f\colon t\mapsto \det{(A\cos{t}+B\sin{t})}\qquad\forall t\in[0,\pi].
\end{equation*}
Since $f(0)=\det{(A)}$ and $f(\pi)=\det{(-A)}=(-1)^3\det{(A)}=-\det{(A)}$, the intermediate value theorem gives that there exists $t_0\in[0,\pi]$ such that $f(t_0)=0$, i.e., the matrix $A\cos{t_0}+B\sin{t_0}$ is singular.
Moreover, this matrix is nonzero since $A$ and $B$ are linearly independent, and so its rank is either $1$ or $2$.
Lemma~\ref{lem.not injective rank 2} then gives that $\mathcal{A}$ is not injective.
\qquad
\end{proof}

As an example, we may run the HMW test on the columns of the following matrix:
\begin{equation}
\label{eq.random cahill example}
\Phi
=\left[
\begin{array}{rrrrrrrr}
2&~~1&~~1&~~0&~~0&~~0&~~1&~~\mathrm{i}\\
-1&0&0&1&1&-1&-2&2\\
0&1&-1&1&-1&2\mathrm{i}&\mathrm{i}&-1
\end{array}
\right].
\end{equation}
In this case, the null space of $\mathbf{A}$ is $1$-dimensional and spanned by a nonsingular matrix.
As such, $\mathcal{A}$ is injective. 
We are now ready to approach the $4M-4$ conjecture in the $M=3$ case:

\begin{theorem}[Theorem~12 in~\cite{BandeiraCMN:13}]
\label{thm.4M4 true 3}
The $4M-4$ conjecture is true when $M=3$.
\end{theorem}

This is proved using the HMW test.
In this case, $4M-4=8=M^2-1$, meaning when $N=4M-4$, the null space of $\mathbf{A}$ should typically be $1$-dimensional.
As such, the null space can be described algebraically in terms of a generalized cross product, and this is leveraged along with the HMW test to construct a real algebraic variety containing all $\mathbf{A}$'s which are not injective; the construction in \eqref{eq.random cahill example} is then used to prove that the complement of this variety is nonempty, thereby proving part (b).

Recently, the American Institute of Mathematics hosted a workshop called ``Frame Theory Intersects Geometry,'' where experts from the two communities discussed various problems, including the $4M-4$ conjecture.
One outcome of this workshop was a paper by Conca, Edidin, Hering and Vinzant~\cite{ConcaEHV:13}, which makes a major stride toward solving the $4M-4$ conjecture:

\begin{theorem}[Theorem~1.1 and Proposition~5.4 in~\cite{ConcaEHV:13}]
\
\begin{itemize}
\item[(a)] Part (a) of the $4M-4$ conjecture is true whenever $M=2^k+1$.
\item[(b)] Part (b) of the $4M-4$ conjecture is true.
\end{itemize}
\end{theorem}

\begin{proof}[sketch]
The proof of (a) uses certain integrality conditions, similar to the proofs of embedding results for complex projective space.
Part (b) is proved using the following basic ideas:
Consider the set of all $M\times N$ complex matrices $\Phi$.
This set has real dimension $2MN$.
The goal is to show that the set of ``bad'' $\Phi$'s (those which fail to yield injectivity) has strictly smaller dimension.
To do this, note from Theorem~\ref{lem.not injective rank 2} that $\Phi=\{\varphi_n\}_{n=1}^N$ is bad precisely when there is an $M\times M$ matrix $Q$ of rank $\leq2$ and Frobenius norm $1$ such that $\varphi_n^*Q\varphi_n=0$ for every $n=1,\ldots,N$.
As such, we lift to the set of pairs $(\Phi,Q)$ which satisfy this relation.
Counting the dimension of this lifted set, we note that the set of $Q$'s has $4M-5$ real dimensions, and for each $Q$ and each $n$, there is a $(2M-1)$-dimensional set of $\varphi_n$'s such that $\varphi_n^*Q\varphi_n=0$.
Thus, the total dimension of bad pairs $(\Phi,Q)$ is $4M-5+(2M-1)N$.
Recall that the bad set we care about is the set of $\Phi$'s for which there exists a $Q$ such that $(\Phi,Q)$ is bad, and so we get our set by projection.
Also, projections never increase the dimension of a set, and so the dimension of our set of bad $\Phi$'s is $\leq 4M-5+(2M-1)N$.
As such, to ensure that this dimension is less than the ambient dimension $2MN$, it suffices to have $4M-5+(2M-1)N<2MN$, or equivalently, $N\geq 4M-4$.
Thus, generic $M\times N$ $\Phi$'s with $N\geq4M-4$ are not bad.
\end{proof}

Note that this result contains the previous cases where $M=2,3$, and the first remaining open case is $M=4$.
On my blog, I offer a US\$100 prize for a proof of the $4M-4$ conjecture, and a can of Coca-Cola for a disproof~\cite{Mixon:13}.

\subsubsection{Minimal Constructions with Injectivity in the Complex Case}

In the absence of a ``good'' characterization of injectivity in the complex case, it is interesting to see explicit minimal constructions, as these shed some insight into the underlying structure of such ensembles.
To this end, there are presently two general constructions, which we describe here.

The construction of Bodmann and Hammen~\cite{BodmannH:13} considers the case where the measurement vectors form a certain harmonic frame.
Specifically, take the $(2M-1)\times(2M-1)$ discrete Fourier transform matrix and collect the first $M$ rows; then the resulting columns are the $M$-dimensional measurement vectors. 
Note that if the original signal is known to be real, then $2M-1$ measurements are injective whenever the measurement vectors are full spark, as this particular harmonic frame is (since all $M\times M$ submatrices are Vandermonde with distinct bases).
Analogously, Bernhard and Hammen exploit the Fejer--Riesz spectral factorization theorem to say that these measurements completely determine signals from another interesting class. 
To be clear, identify a vector $(c_m)_{m=0}^{M-1}$ with the polynomial $\sum_{m=0}^{M-1}c_mz^m$; then Bernhard and Hammen uniquely determine the vector if the roots of the corresponding polynomial all lie outside the open unit disk in the complex plane. 
In general, they actually say the polynomial is one of $2^M$ possibilities; here, the only ambiguity is whether a given root is at $z_m$ or $1/\overline{z_m}$, that is, we can flip any root from outside to inside the disk.
Note that this is precisely how much ambiguity we have in the real case after taking only $M$ measurements, and in that case, we know it suffices to take only $M-1$ additional measurements.
Next, in addition to taking these $2M-1$ measurements from before (viewed as equally spaced points on the complex unit circle), they also take $2M-1$ measurements corresponding to equally spaced points on another unit circle in the complex plane, this one being the image of the real line under a specially chosen (think ``sufficiently irrational'') Cayley map.
However, this makes a total of $4M-2$ measurements, whereas the goal is to find $4M-4$ injective measurements -- to fix this, they actually pick the second circle in such a way that it intersects the first at two points, and that these intersection points correspond to measurements from both circles, so we might as well throw two of them away.

The second known construction is due to Fickus, Nelson, Wang and the author~\cite{FickusMNW:13}.
Here, we apply two main ideas: (1) a signal's intensity measurements with the Fourier transform is the Fourier transform of the signal's autocorrelation; and (2) if a real, even function is sufficiently zero-padded, then it can be recovered (up to a global sign factor) from its autocorrelation.
(Verifying (2) in small dimensions is a fun exercise.)
In~\cite{FickusMNW:13}, we show how to generalize (2) to completely determine zero-padded complex functions, and then we identify these autocorrelations as the inverse Fourier transforms of intensity measurements with $4M-2$ truncated and modulated discrete cosine functions.
At this point, we identify certain redundancies in our intensity measurements -- two of them are completely determined by the others, so we can remove them.

Considering these minimal constructions, it is striking that they both come from a construction of size $4M-2$.
This begs the following question:
Does every injective ensemble of size $4M-2$ contain an injective ensemble of size $4M-4$?

\section{Almost Injectivity}

In both the real and complex cases, there appears to be a phase transition which governs how many intensity measurements are necessary and generically sufficient for injectivity.
Interestingly, one can save a factor of $2$ in this number of measurements by slightly weakening the desired notion of injectivity~\cite{BalanCE:06,FlammiaSC:05}.
To be explicit, we start with the following definition:

\begin{definition}
\label{def.almost injective}
The intensity measurement mapping $\mathcal{A}$ is said to be \emph{almost injective} if $\mathcal{A}^{-1}(\mathcal{A}(x))=\{\omega x:|\omega|=1\}$ for every $x$ in an open, dense subset of $\mathbb{F}^M$.
$\operatorname{AlmInj}[\Phi,\mathbb{F}^{M\times N}]$ denotes the statement that the intensity measurement mapping $\mathcal{A}$ associated with $\Phi$ is almost injective.
\end{definition}

This section studies the phase transition for almost injectivity.
Much like injectivity, we have a much better understanding of the real case than the complex case, and we consider these separately.

\subsection{Almost Injectivity in the Real Case}

In this section, we start by characterizing ensembles of measurement vectors which yield almost injective intensity measurements, and similar to the characterization of injectivity, the basic idea behind the analysis is to consider sums and differences of signals with identical intensity measurements.
Our characterization starts with the following lemma:

\begin{theorem}[Lemma~9 in~\cite{FickusMNW:13}]
\label{lem.almost injective and Minkowski sum}
Consider $\Phi=\{\varphi_n\}_{n=1}^N\subseteq\mathbb{R}^M$ and the intensity measurement mapping $\mathcal{A}\colon\mathbb{R}^M/\{\pm1\}\rightarrow\mathbb{R}^N$ defined by $(\mathcal{A}(x))(n):=|\langle x,\varphi_n\rangle|^2$. 
Then $\mathcal{A}$ is almost injective if and only if almost every $x\in\mathbb{R}^M$ is not in the Minkowski sum
$\operatorname{span}(\Phi_{S})^\perp\setminus\{0\}+\operatorname{span}(\Phi_{S^\mathrm{c}})^\perp\setminus\{0\}$
for all $S\subseteq\{1,\ldots,N\}$. 
More precisely, $\mathcal{A}^{-1}(\mathcal{A}(x))=\{\pm x\}$ if and only if $x\notin\operatorname{span}(\Phi_{S})^\perp\setminus\{0\}+\operatorname{span}(\Phi_{S^\mathrm{c}})^\perp\setminus\{0\}$ for any $S\subseteq\{1,\ldots,N\}$.
\end{theorem}

\begin{proof}
By the definition of the mapping $\mathcal{A}$, for $x,y\in\mathbb{R}^M$ we have $\mathcal{A}(x)=\mathcal{A}(y)$ if and only if $|\langle x,\varphi_n\rangle|=|\langle y,\varphi_n\rangle|$ for all $n\in\{1,\ldots,N\}$.
This occurs precisely when there is a subset $S\subseteq\{1,\ldots,N\}$ such that $\langle x,\varphi_n\rangle=-\langle y,\varphi_n\rangle$ for every $n\in S$ and $\langle x,\varphi_n\rangle=\langle y,\varphi_n\rangle$ for every $n\in S^\mathrm{c}$.
Thus, $\mathcal{A}^{-1}(\mathcal{A}(x))=\{\pm x\}$ if and only if for every $y\neq\pm x$ and for every $S\subseteq\{1,\ldots,N\}$, either there exists an $n\in S$ such that $\langle x+y,\varphi_n\rangle\neq0$ or an $n\in S^\mathrm{c}$ such that $\langle x-y,\varphi_n\rangle\neq0$.
We claim that this occurs if and only if $x$ is not in the Minkowski sum
$\operatorname{span}(\Phi_{S})^\perp\setminus\{0\}+\operatorname{span}(\Phi_{S^\mathrm{c}})^\perp\setminus\{0\}$
for all $S\subseteq\{1,\ldots,N\}$, which would complete the proof.
We verify the claim by seeking the contrapositive in each direction.

$(\Rightarrow)$ Suppose
$x\in\operatorname{span}(\Phi_{S})^\perp\setminus\{0\}+\operatorname{span}(\Phi_{S^\mathrm{c}})^\perp\setminus\{0\}$.
Then there exists $u\in\operatorname{span}(\Phi_{S})^\perp\setminus\{0\}$ and $v\in\operatorname{span}(\Phi_{S^\mathrm{c}})^\perp\setminus\{0\}$ such that $x=u+v$.
Taking $y:=u-v$, we see that
$x+y=2u\in\operatorname{span}(\Phi_{S})^\perp\setminus\{0\}$
and
$x-y=2v\in\operatorname{span}(\Phi_{S^\mathrm{c}})^\perp\setminus\{0\}$,
which means that there is no $n\in S$ such that $\langle x+y,\varphi_n\rangle\neq0$ nor $n\in S^\mathrm{c}$ such that $\langle x-y,\varphi_n\rangle\neq0$.
Furthermore, $u$ and $v$ are nonzero, and so $y\neq\pm x$.

$(\Leftarrow)$ Suppose $y\neq\pm x$ and for every $S\subseteq\{1,\ldots,N\}$ there is no $n\in S$ such that $\langle x+y,\varphi_n\rangle\neq0$ nor $n\in S^\mathrm{c}$ such that $\langle x-y,\varphi_n\rangle\neq0$.
Then $x+y\in\operatorname{span}(\Phi_{S})^\perp\setminus\{0\}$ and $x-y\in\operatorname{span}(\Phi_{S^\mathrm{c}})^\perp\setminus\{0\}$.
Since
$x=\frac{1}{2}(x+y)+\frac{1}{2}(x-y)$,
we have that $x\in\operatorname{span}(\Phi_{S})^\perp\setminus\{0\}+\operatorname{span}(\Phi_{S^\mathrm{c}})^\perp\setminus\{0\}$.
\end{proof}

The above characterization can be simplified to form the following partial characterization of almost injectivity:

\begin{theorem}[Theorem~10 in~\cite{FickusMNW:13}]
\label{thm.almost injective and Minkowski sum proper subspace}
Consider $\Phi=\{\varphi_n\}_{n=1}^N\subseteq\mathbb{R}^M$ and the intensity measurement mapping $\mathcal{A}\colon\mathbb{R}^M/\{\pm1\}\rightarrow\mathbb{R}^N$ defined by $(\mathcal{A}(x))(n):=|\langle x,\varphi_n\rangle|^2$. 
Suppose $\Phi$ spans $\mathbb{R}^M$ and each $\varphi_n$ is nonzero.
Then $\mathcal{A}$ is almost injective if and only if the Minkowski sum
$\operatorname{span}(\Phi_{S})^\perp+\operatorname{span}(\Phi_{S^\mathrm{c}})^\perp$
is a proper subspace of $\mathbb{R}^M$ for each nonempty proper subset $S\subseteq\{1,\ldots,N\}$.
\end{theorem}

Note that the above result is not terribly surprising considering Theorem~\ref{lem.almost injective and Minkowski sum}, as the new condition involves a simpler Minkowski sum in exchange for additional (reasonable and testable) assumptions on $\Phi$.
The proof of this theorem amounts to measuring the difference between the two Minkowski sums:

\begin{proof}[Proof of Theorem~\ref{thm.almost injective and Minkowski sum proper subspace}]
First note that the spanning assumption on $\Phi$ implies
\begin{equation*}
\operatorname{span}(\Phi_S)^\perp\cap\operatorname{span}(\Phi_{S^\mathrm{c}})^\perp
=\bigl(\operatorname{span}(\Phi_S)+\operatorname{span}(\Phi_{S^\mathrm{c}})\bigr)^\perp
=\operatorname{span}(\Phi)^\perp
=\{0\},
\end{equation*}
and so one can prove the following identity:
\begin{align}
\nonumber
&\operatorname{span}(\Phi_{S})^\perp\setminus\{0\}+\operatorname{span}(\Phi_{S^\mathrm{c}})^\perp\setminus\{0\}\\
\label{eq.Minkowski sum equality}
&\qquad=\left(\operatorname{span}(\Phi_{S})^\perp+\operatorname{span}(\Phi_{S^\mathrm{c}})^\perp\right)
\setminus\left(\operatorname{span}(\Phi_{S})^\perp\cup\operatorname{span}(\Phi_{S^\mathrm{c}})^\perp\right).
\end{align}
From Theorem~\ref{lem.almost injective and Minkowski sum} we know that $\mathcal{A}$ is almost injective if and only if almost every $x\in\mathbb{R}^M$ is not in the Minkowski sum
$\operatorname{span}(\Phi_{S})^\perp\setminus\{0\}+\operatorname{span}(\Phi_{S^\mathrm{c}})^\perp\setminus\{0\}$
for any $S\subseteq\{1,\ldots,N\}$.
In other words, the Lebesgue measure (which we denote by $\operatorname{Leb}[\cdot]$) of this Minkowski sum is zero for each $S\subseteq\{1,\ldots,N\}$.
By \eqref{eq.Minkowski sum equality}, this equivalently means that the Lebesgue measure of
$\left(\operatorname{span}(\Phi_{S})^\perp+\operatorname{span}(\Phi_{S^\mathrm{c}})^\perp\right)
\setminus\left(\operatorname{span}(\Phi_{S})^\perp\cup\operatorname{span}(\Phi_{S^\mathrm{c}})^\perp\right)$
is zero for each $S\subseteq\{1,\ldots,N\}$.
Since $\Phi$ spans $\mathbb{R}^M$, this set is empty (and therefore has Lebesgue measure zero) when $S=\emptyset$ or $S=\{1,\ldots,N\}$.
Also, since each $\varphi_n$ is nonzero, we know that $\operatorname{span}(\Phi_{S})^\perp$ and $\operatorname{span}(\Phi_{S^\mathrm{c}})^\perp$ are proper subspaces of $\mathbb{R}^M$ whenever $S$ is a nonempty proper subset of $\{1,\ldots,N\}$, and so in these cases both subspaces must have Lebesgue measure zero.
As such, we have that for every nonempty proper subset $S\subseteq\{1,\ldots,N\}$,
\begin{align*}
&\operatorname{Leb}\left[\left(\operatorname{span}(\Phi_{S})^\perp+\operatorname{span}(\Phi_{S^\mathrm{c}})^\perp\right)
\setminus\left(\operatorname{span}(\Phi_{S})^\perp\cup\operatorname{span}(\Phi_{S^\mathrm{c}})^\perp\right)\right]\\
&\quad\geq\operatorname{Leb}\left[\operatorname{span}(\Phi_{S})^\perp+\operatorname{span}(\Phi_{S^\mathrm{c}})^\perp\right]
-\operatorname{Leb}\left[\operatorname{span}(\Phi_{S})^\perp\right]-\operatorname{Leb}\left[\operatorname{span}(\Phi_{S^\mathrm{c}})^\perp\right]\\
&\quad=\operatorname{Leb}\left[\operatorname{span}(\Phi_{S})^\perp+\operatorname{span}(\Phi_{S^\mathrm{c}})^\perp\right]\\
&\quad\geq\operatorname{Leb}\left[\left(\operatorname{span}(\Phi_{S})^\perp+\operatorname{span}(\Phi_{S^\mathrm{c}})^\perp\right)\setminus\left(\operatorname{span}(\Phi_{S})^\perp\cup\operatorname{span}(\Phi_{S^\mathrm{c}})^\perp\right)\right].
\end{align*}
In summary, $\left(\operatorname{span}(\Phi_{S})^\perp+\operatorname{span}(\Phi_{S^\mathrm{c}})^\perp\right)
\setminus\left(\operatorname{span}(\Phi_{S})^\perp\cup\operatorname{span}(\Phi_{S^\mathrm{c}})^\perp\right)$
having Lebesgue measure zero for each $S\subseteq\{1,\ldots,N\}$ is equivalent to $\operatorname{span}(\Phi_{S})^\perp+\operatorname{span}(\Phi_{S^\mathrm{c}})^\perp$ having Lebesgue measure zero for each nonempty proper subset $S\subseteq\{1,\ldots,N\}$, which in turn is equivalent to the Minkowski sum $\operatorname{span}(\Phi_{S})^\perp+\operatorname{span}(\Phi_{S^\mathrm{c}})^\perp$ being a proper subspace of $\mathbb{R}^M$ for each nonempty proper subset $S\subseteq\{1,\ldots,N\}$, as desired.
\end{proof}

At this point, consider the following stronger restatement of Theorem~\ref{thm.almost injective and Minkowski sum proper subspace}: 
``Suppose each $\varphi_n$ is nonzero.
Then $\mathcal{A}$ is almost injective if and only if $\Phi$ spans $\mathbb{R}^M$ and the Minkowski sum
$\operatorname{span}(\Phi_{S})^\perp+\operatorname{span}(\Phi_{S^\mathrm{c}})^\perp$
is a proper subspace of $\mathbb{R}^M$ for each nonempty proper subset $S\subseteq\{1,\ldots,N\}$.''
Note that we can move the spanning assumption into the condition because if $\Phi$ does not span, then we can decompose almost every $x\in\mathbb{R}^M$ as $x=u+v$ such that $u\in\operatorname{span}(\Phi)$ and $v\in\operatorname{span}(\Phi)^\perp$ with $v\neq0$, and defining $y:=u-v$ then gives $\mathcal{A}(y)=\mathcal{A}(x)$ despite the fact that $y\neq\pm x$.
As for the assumption that the $\varphi_n$'s are nonzero, we note that having $\varphi_n=0$ amounts to having the $n$th entry of $\mathcal{A}(x)$ be zero for all $x$.
As such, $\Phi$ yields almost injectivity precisely when the nonzero members of $\Phi$ together yield almost injectivity.
With this identification, the stronger restatement of Theorem~\ref{thm.almost injective and Minkowski sum proper subspace} above can be viewed as a complete characterization of almost injectivity.
Next, we will replace the Minkowski sum condition with a rather elegant condition involving the ranks of $\Phi_S$ and $\Phi_{S^\mathrm{c}}$:

\begin{theorem}[Theorem~11 in~\cite{FickusMNW:13}]
\label{thm.almost injective and sum rank >M}
Consider $\Phi=\{\varphi_n\}_{n=1}^N\subseteq\mathbb{R}^M$ and the intensity measurement mapping $\mathcal{A}\colon\mathbb{R}^M/\{\pm1\}\rightarrow\mathbb{R}^N$ defined by $(\mathcal{A}(x))(n):=|\langle x,\varphi_n\rangle|^2$. 
Suppose each $\varphi_n$ is nonzero.
Then $\mathcal{A}$ is almost injective if and only if $\Phi$ spans $\mathbb{R}^M$ and $\operatorname{rank}\Phi_{S}+\operatorname{rank}\Phi_{S^\mathrm{c}}>M$
for each nonempty proper subset $S\subseteq\{1,\ldots,N\}$.
\end{theorem}

\begin{proof}
Considering the discussion after the proof of Theorem~\ref{thm.almost injective and Minkowski sum proper subspace}, it suffices to assume that $\Phi$ spans $\mathbb{R}^M$.
Furthermore, considering Theorem~\ref{thm.almost injective and Minkowski sum proper subspace}, it suffices to characterize when $\dim\left(\operatorname{span}(\Phi_{S})^\perp+\operatorname{span}(\Phi_{S^\mathrm{c}})^\perp\right)<M$.
By the inclusion-exclusion principle for subspaces, we have
\begin{align*}
&\dim\left(\operatorname{span}(\Phi_{S})^\perp+\operatorname{span}(\Phi_{S^\mathrm{c}})^\perp\right)\\
&=\dim\left(\operatorname{span}(\Phi_{S})^\perp\right)
+\dim\left(\operatorname{span}(\Phi_{S^\mathrm{c}})^\perp\right)
-\dim\left(\operatorname{span}(\Phi_{S})^\perp\cap\operatorname{span}(\Phi_{S^\mathrm{c}})^\perp\right).
\end{align*}
Since $\Phi$ is assumed to span $\mathbb{R}^M$, we also have that
$\operatorname{span}(\Phi_{S})^\perp\cap\operatorname{span}(\Phi_{S^\mathrm{c}})^\perp=\{0\}$,
and so
\begin{align*}
&\dim\left(\operatorname{span}(\Phi_{S})^\perp+\operatorname{span}(\Phi_{S^\mathrm{c}})^\perp\right)\\
&\quad=\Big(M-\dim\left(\operatorname{span}(\Phi_{S})\right)\Big)
+\Big(M-\dim\left(\operatorname{span}(\Phi_{S^\mathrm{c}})\right)\Big)
-0\\
&\quad=2M-\operatorname{rank}\Phi_{S}-\operatorname{rank}\Phi_{S^\mathrm{c}}.
\end{align*}
As such, $\dim\left(\operatorname{span}(\Phi_{S})^\perp+\operatorname{span}(\Phi_{S^\mathrm{c}})^\perp\right)<M$ precisely when $\operatorname{rank}\Phi_{S}+\operatorname{rank}\Phi_{S^\mathrm{c}}>M$.
\end{proof}

At this point, we point out some interesting consequences of Theorem~\ref{thm.almost injective and sum rank >M}.
First of all, $\Phi$ cannot be almost injective if $N<M+1$ since $\operatorname{rank}\Phi_{S}+\operatorname{rank}\Phi_{S^\mathrm{c}}\leq|S|+|S^\mathrm{c}|=N$.
Also, in the case where $N=M+1$, we note that $\Phi$ is almost injective precisely when $\Phi$ is full spark, that is, every size-$M$ subcollection is a spanning set (note this implies that all of the $\varphi_n$'s are nonzero).
In fact, every full spark $\Phi$ with $N\geq M+1$ yields almost injective intensity measurements, which in turn implies that a generic $\Phi$ yields almost injectivity when $N\geq M+1$~\cite{BalanCE:06}.
This is in direct analogy with injectivity in the real case; here, injectivity requires $N\geq 2M-1$, injectivity with $N=2M-1$ is equivalent to being full spark, and being full spark suffices for injectivity whenever $N\geq 2M-1$~\cite{BalanCE:06}.
Another thing to check is that the condition for injectivity implies the condition for almost injectivity (it does).
Overall, we have the following phase transition result:

\begin{theorem}[essentially proved in~\cite{BalanCE:06}]
$\operatorname{AlmInj}[\Phi,\mathbb{R}^{M\times N}]$ exhibits a phase transition at $N=M+1$.
\end{theorem}

Having established that full spark ensembles of size $N\geq M+1$ yield almost injective intensity measurements, we note that checking whether a matrix is full spark is $\mathrm{NP}$-hard in general~\cite{Khachiyan:95}.
Granted, there are a few explicit constructions of full spark ensembles which can be used~\cite{AlexeevCM:12,PuschelK:05}, but it would be nice to have a condition which is not computationally difficult to test in general.
We provide one such condition in the following theorem, but first, we briefly review the requisite frame theory.

A \emph{frame} is an ensemble $\Phi=\{\varphi_n\}_{n=1}^N\subseteq\mathbb{R}^M$ together with \emph{frame bounds} $0<A\leq B<\infty$ with the property that for every $x\in\mathbb{R}^M$,
\begin{equation*}
A\|x\|^2
\leq\sum_{n=1}^N|\langle x,\varphi_n\rangle|^2
\leq B\|x\|^2.
\end{equation*}
When $A=B$, the frame is said to be \emph{tight}, and such frames come with a painless reconstruction formula:
\begin{equation*}
x=\frac{1}{A}\sum_{n=1}^N\langle x,\varphi_n\rangle\varphi_n.
\end{equation*}
To be clear, the theory of frames originated in the context of infinite-dimensional Hilbert spaces~\cite{DaubechiesGM:86,DuffinS:52}, and frames have since been studied in finite-dimensional settings, primarily because this is the setting in which they are applied computationally.
Of particular interest are so-called \emph{unit norm tight frames (UNTFs)}, which are tight frames whose frame elements have unit norm: $\|\varphi_n\|=1$ for every $n=1,\ldots,N$.
Such frames are useful in applications; for example, if one encodes a signal $x$ using frame coefficients $\langle x,\varphi_n\rangle$ and transmits these coefficients across a channel, then UNTFs are optimally robust to noise~\cite{GoyalVT:98} and one erasure~\cite{CasazzaK:03}.
Intuitively, this optimality comes from the fact that frame elements of a UNTF are particularly well-distributed in the unit sphere~\cite{BenedettoF:03}.
Another pleasant feature of UNTFs is that it is straightforward to test whether a given frame is a UNTF:
Letting $\Phi=[\varphi_1\cdots\varphi_N]$ denote an $M\times N$ matrix whose columns are the frame elements, then $\Phi$ is a UNTF precisely when each of the following occurs simultaneously:
\begin{itemize}
\item[(a)] the rows have equal norm
\item[(b)] the rows are orthogonal
\item[(c)] the columns have unit norm
\end{itemize}
(This is a direct consequence of the tight frame's reconstruction formula and the fact that a UNTF has unit-norm frame elements; furthermore, since the columns have unit norm, it is not difficult to see that the rows will necessarily have norm $\sqrt{N/M}$.)
In addition to being able to test that an ensemble is a UNTF, various UNTFs can be constructed using \emph{spectral tetris}~\cite{CasazzaFMWZ:11} (though such frames necessarily have $N\geq 2M$), and \emph{every} UNTF can be constructed using the recent theory of \emph{eigensteps}~\cite{CahillFMPS:13,FickusMPS:13}.
Now that UNTFs have been properly introduced, we relate them to almost injectivity for phase retrieval:

\begin{theorem}[Theorem~12 in~\cite{FickusMNW:13}]
\label{thm.almost injectivity and relatively prime UNTFs}
If $M$ and $N$ are relatively prime, then every unit norm tight frame $\Phi=\{\varphi_n\}_{n=1}^N\subseteq\mathbb{R}^M$ yields almost injective intensity measurements.
\end{theorem}

\begin{proof}
Pick a nonempty proper subset $S\subseteq\{1,\ldots,N\}$.
By Theorem~\ref{thm.almost injective and sum rank >M}, it suffices to show that
$\operatorname{rank}\Phi_{S}+\operatorname{rank}\Phi_{S^\mathrm{c}}>M$,
or equivalently,
$\operatorname{rank}\Phi_{S}\Phi_{S}^*+\operatorname{rank}\Phi_{S^\mathrm{c}}\Phi_{S^{\mathrm{c}}}^*>M$.
Note that since $\Phi$ is a unit norm tight frame, we also have
\begin{equation*}
\Phi_{S}\Phi_{S}^*+\Phi_{S^\mathrm{c}}\Phi_{S^{\mathrm{c}}}^*=\Phi\Phi^*=\tfrac{N}{M}I,
\end{equation*}
and so $\Phi_{S}\Phi_{S}^*$ and $\Phi_{S^\mathrm{c}}\Phi_{S^{\mathrm{c}}}^*$ are simultaneously diagonalizable, i.e., there exists a unitary matrix $U$ and diagonal matrices $D_1$ and $D_2$ such that
\begin{align*}
UD_1U^*+UD_2U^*=\Phi_{S}\Phi_{S}^*+\Phi_{S^\mathrm{c}}\Phi_{S^{\mathrm{c}}}^*=\tfrac{N}{M}I.
\end{align*}
Conjugating by $U^*$, this then implies that $D_1+D_2=\tfrac{N}{M}I$.
Let $L_1\subseteq\{1,\ldots,M\}$ denote the diagonal locations of the nonzero entries in $D_1$, and $L_2\subseteq\{1,\ldots,M\}$ similarly for $D_2$.
To complete the proof, we need to show that $|L_1|+|L_2|>M$ (since $|L_1|+|L_2|=\operatorname{rank}\Phi_{S}\Phi_{S}^*+\operatorname{rank}\Phi_{S^\mathrm{c}}\Phi_{S^{\mathrm{c}}}^*$).
Note that $L_1\cup L_2\neq\{1,\ldots,M\}$ would imply that $D_1+D_2$ has at least one zero in its diagonal, contradicting the fact that $D_1+D_2$ is a nonzero multiple of the identity; as such, $L_1\cup L_2=\{1,\ldots,M\}$ and $|L_1|+|L_2|\geq M$.
We claim that this inequality is strict due to the assumption that $M$ and $N$ are relatively prime.
To see this, it suffices to show that $L_1\cap L_2$ is nonempty.
Suppose to the contrary that $L_1$ and $L_2$ are disjoint.
Then since $D_1+D_2=\tfrac{N}{M}I$, every nonzero entry in $D_1$ must be $N/M$.
Since $S$ is a nonempty proper subset of $\{1,\ldots,N\}$, this means that there exists $K\in(0,M)$ such that $D_1$ has $K$ entries which are $N/M$ and $M-K$ which are $0$.
Thus,
\begin{equation*}
|S|
=\operatorname{Tr}[\Phi_{S}^*\Phi_{S}]
=\operatorname{Tr}[\Phi_{S}\Phi_{S}^*]
=\operatorname{Tr}[UD_1U^*]
=\operatorname{Tr}[D_1]
=K(N/M),
\end{equation*}
implying that $N/M=|S|/K$ with $K\neq M$ and $|S|\neq N$.
Since this contradicts the assumption that $N/M$ is in lowest form, we have the desired result.
\end{proof}

In general, whether a UNTF $\Phi$ yields almost injective intensity measurements is determined by whether it is \emph{orthogonally partitionable}: $\Phi$ is orthogonally partitionable if there exists a partition $S\sqcup S^\mathrm{c}=\{1,\ldots,N\}$ such that $\operatorname{span}(\Phi_S)$ is orthogonal to $\operatorname{span}(\Phi_{S^\mathrm{c}})$.
Specifically, a UNTF yields almost injective intensity measurements precisely when it is not orthogonally partitionable.
Historically, this property of UNTFs has been pivotal to the understanding of singularities in the algebraic variety of UNTFs~\cite{DykemaS:06}, and it has also played a key role in solutions to the Paulsen problem~\cite{BodmannC:10,CasazzaFM:12}.
However, it is not clear in general how to efficiently test for this property; this is why Theorem~\ref{thm.almost injectivity and relatively prime UNTFs} is so powerful.

\subsection{Almost Injectivity in the Complex Case}

The complex case is not understood nearly as well as the real case, but the phase transition is arguably better understood than the one for injectivity in the complex case.
However, almost injectivity hasn't received as much attention, so there are no known characterizations in the complex case, let alone ``useful'' ones.
To begin our discussion of the phase transition, we consider the following lemma (the proof is enjoyable):

\begin{theorem}
\label{theorem.differential geometry bound for almost injectivity}
Suppose $\mathcal{A}\colon\mathbb{R}^P\rightarrow\mathbb{R}^N$ has a continuous Jacobian $J$ over some open set $U\subseteq\mathbb{R}^P$.
If $\operatorname{rank}(J(x))<P$ for every $x\in U$, then $\mathcal{A}$ is not injective when restricted to $U$.
\end{theorem}

\begin{proof}
Let $z$ be a point in $U$ which maximizes $\operatorname{rank}(J(x))$, and let $K$ denote the rank of $J(z)$.
Then there are $K$ linearly independent columns of $J(z)$ forming the submatrix $J_\mathcal{K}(z)$.
Furthermore, these columns remain linearly independent in any sufficiently small neighborhood $B$ of $z$ in $U$.
(This can be established using the continuous mapping $x\mapsto \det[(J_\mathcal{K}(x))^*J_\mathcal{K}(x)]$.)
As such, we can define the continuous mapping
\begin{equation*}
x\mapsto P(x):=I-J_\mathcal{K}(x)[(J_\mathcal{K}(x))^*J_\mathcal{K}(x)]^{-1}(J_\mathcal{K}(x))^*
\end{equation*}
for all $x\in B$.
By construction, $P(x)$ is the orthogonal projection onto the null space of $J(x)$.
Pick some nonzero member $v$ of the null space of $J(z)$, and consider the continuous vector field $x\mapsto P(x)v$ over $B$.
By the Peano existence theorem, there exists $\epsilon>0$ and $\gamma\colon[0,\epsilon]\rightarrow B$ such that $\gamma(0)=z$ and $\gamma'(t)=P(\gamma(t))v$ for every $t\in[0,\epsilon]$.
Since $\gamma'(t)=P(\gamma(t))v$ is in the null space of $J(\gamma(t))$, we then have
\begin{equation*}
0=J(\gamma(t))\gamma'(t)=\frac{d}{dt}\Big(\mathcal{A}(\gamma(t))\Big)
\end{equation*}
for every $t\in[0,\epsilon]$, meaning $\mathcal{A}(x)$ is constant over all $x\in\gamma([0,\epsilon])$.
Furthermore, $\gamma([0,\epsilon])$ contains more than a single point, since otherwise $\gamma$ is constant, contradicting $\gamma'(0)=v\neq0$.
As such, $\mathcal{A}$ is not injective over any sufficiently small neighborhood of $z$, let alone $U$.
\end{proof}

Take an almost injective intensity measurement mapping $\mathcal{A}$ and restrict it to an open set $S$ of $x$'s for which $\mathcal{A}^{-1}(\mathcal{A}(x))=\{\omega x:|\omega|=1\}$.
Note that $\mathcal{A}$ is injective over $S$ by assumption.
Considering $(\mathbb{C}^M\setminus\{0\})/\mathbb{T}$ is a smooth manifold of real dimension $P=2M-1$, we can intersect $S$ with a patch to get an open set $U$, and consider the Jacobian of $\mathcal{A}$ in the patch's local coordinates.
By the contrapositive of Theorem~\ref{theorem.differential geometry bound for almost injectivity}, we have that $N\geq\operatorname{rank}(J(x))\geq P=2M-1$ for some $x\in U$.
This may lead one to believe that $\operatorname{AlmInj}[\Phi,\mathbb{C}^{M\times N}]$ exhibits a phase transition at $N=2M-1$, but there is evidence to suggest that is this off by $1$ due to algebraic properties of intensity measurements:

\begin{conjecture}
$\operatorname{AlmInj}[\Phi,\mathbb{C}^{M\times N}]$ exhibits a phase transition at $N=2M$.
\end{conjecture}

To be clear, part (b) of this conjecture was proved by Balan, Casazza and Edidin~\cite{BalanCE:06}, whereas a sketch of the proof of part (a) is provided in~\cite{FlammiaSC:05}.
However, the latter sketch leaves much to be desired -- while the argument is believable in principle, it is unclear whether their use of real algebraic geometry is sufficiently rigorous.
For explicit minimal constructions in this case (assuming the conjecture is true), see~\cite{Finkelstein:04,FlammiaSC:05}.

\section*{Acknowledgments}
The author was supported by NSF Grant No.\ DMS-1321779.
The views expressed in this chapter are those of the author and do not reflect the official policy or position of the United States Air Force, Department of Defense, or the U.S.\ Government.

\end{document}